%% file: isit_noisy.tex
\documentclass[onecolumn]{IEEEtran}
\usepackage[cmex10]{amsmath} % ensure not too small fonts (so Type 1 still)
\usepackage{array}
\usepackage{threeparttable}
\usepackage{soul}
\usepackage{multirow}

\input{format.tex}

\newcommand\be{\text{Bern}}

\crefname{Appendix}{Appendix}{Appendices}

\begin{document}
\sloppy

\title{Estimation of Sparsity via Simple Measurements}
\author{\IEEEauthorblockN{Abhishek Agarwal}\hspace*{1in}
\IEEEauthorblockN{Larkin Flodin}\hspace*{1in}
\IEEEauthorblockN{Arya Mazumdar}}

\allowdisplaybreaks
\sloppy
\maketitle

{\renewcommand{\thefootnote}{}\footnotetext{

\vspace{-.2in}
 
\noindent\rule{1.5in}{.4pt}

College of Information and Computer Sciences, University of Massachusetts Amherst. 
\texttt{\{abhiag,lflodin,arya\}@cs.umass.edu}. This work was supported by NSF CCF Awards 1318093, 1642658,  1642550.
}
\renewcommand{\thefootnote}{\arabic{footnote}}
\setcounter{footnote}{0}

\begin{abstract}
We consider several related problems of estimating the `sparsity' or number of nonzero elements $d$ in a length $n$ vector $\bfx$ by observing only $\bfb = M \odot \bfx$, where $M$ is a predesigned test matrix independent of $\bfx$, and the operation $\odot$ varies between problems. We aim to provide a $\Delta$-approximation of sparsity for some constant $\Delta$ with a minimal number of measurements (rows of $M$). This framework generalizes multiple problems, such as estimation of sparsity in group testing and compressed sensing. We use techniques from coding theory as well as probabilistic methods to show that $O(D \log D \log n)$ rows are sufficient when the operation $\odot$ is logical OR (i.e., group testing), and nearly this many are necessary, where $D$ is a known upper bound on $d$. When instead the operation $\odot$ is multiplication over $\real$ or a finite field $\F_q$, we show that respectively $\Theta(D)$ and $\Theta(D \log_q \frac{n}{D})$ measurements are necessary and sufficient.
\end{abstract}

\section{Introduction}
\label{sec:intro}

Suppose that we want to identify an $n$ dimensional vector $\bfx \in \F^n$, however, we can only observe the output $\bfb$ where
\begin{equation}\label{measurement}
M \odot \bfx = \bfb
\end{equation}
for a designed matrix $M \in {\F}^{m\times n}$. Let $M_{ij}$ denote the $(i,j)$th entry of $M$ and let $\bfx_i$ denote the $i$th component of $\bfx$. We will frequently refer to a single row of $M$ as a ``test'' or ``measurement.'' If we define the operation $\odot$ in \cref{measurement} as standard matrix multiplication over the field, and $\F = \real$, the problem of identifying $\bfx$ is known as compressed sensing (and solvable with $m \ll n$ tests when $\bfx$ is `sparse'). If instead we define the operation $\odot$ as the logical OR, so that $\bfb_i \defeq \bigvee_{j : M_{ij}=1} \bfx_j$, and $\F = \F_2$, the identification problem is known as group testing. We can easily identify $\bfx$ when $m \geq n$, for example by taking $M$ to be the identity matrix. Let $d$ be the {\em sparsity} or number of nonzero entries of the vector $\bfx$. For the case when it is known a priori that $\bfx$ is sparse (that is to say $d \ll n$), it has been shown that $m = O(d^2 \log{n})$ measurements are sufficient for identification in the group testing setting, and even fewer measurements, $m =2d$, are necessary and sufficient in the compressed sensing setting.
The group testing result is only a $\log d$ factor away from the known lower bound \cite{du1999combinatorial}.

Without any information about $\bfx$, it would be desirable to first estimate $d$, and then use this estimate to choose a suitable strategy to identify the $d$-sparse vector $\bfx$. A considerable body of research (see surveys in \cite{gilbert2010sparse,du1999combinatorial}) has focused on the identification problem when an upper bound on $d$ is known in advance. In comparison less work has been done on the problem of simply estimating $d$, without trying to determine which specific entries are nonzero \cite{damaschke2010bounds,damaschke2010competitive,cheng2011efficient,cheng2014efficient,ron2014power,acharya2015adaptive}. For group testing in the adaptive setting (when each subsequent test can depend on the results of previous tests), a recent result of Falahatgar et al. \cite{falahatgar2016estimating} allows approximation of $d$ in as few as $O(\log \log d)$ tests, though with a small probability of error.
In this paper, we solely concentrate on the non-adaptive version, in which the entire test matrix is specified before the results of any tests are seen.

In particular, we provide tight upper and lower bounds on the number of measurements required by any deterministic algorithm for estimating $d$ within a constant multiplicative factor of $\Delta$ ($\Delta$-approximation) in three different settings of this problem (i.e., definitions of the operation $\odot$  and field $\F$ in \cref{measurement}). Note that $\Delta$-approximation implies an estimate $\hat{d}$ such that
\begin{equation}
\label{estimation_crit}
\frac{1}{\Delta} \leq \frac{\hat{d}}{d} \leq \Delta.
\end{equation}

Earlier results by Damaschke and Muhammad \cite{damaschke2010bounds,damaschke2010competitive} in the non-adaptive group testing setting show that when no upper bound on the number of defectives is known, $O(\log n)$ queries are needed to approximate $d$ even when the scheme is allowed to fail for a small number of inputs, and that any deterministic strategy capable of exactly determining $d$ requires enough information to reconstruct the vector $\bfx$ exactly. In the group testing model, we restrict our attention to the unstudied problem of bounds for non-adaptive approximation schemes which work for all inputs, i.e., those that always produce an estimate of $d$ within the specified range of allowable estimates.

The majority of work in the area of compressed  sensing  is concerned with the more difficult problem of how to proceed when $\bfx$ is not exactly sparse, but instead is approximately sparse, meaning it is close in $\ell_2$ norm to a sparse vector. To our knowledge all such works are concerned with recovery of the vector $\bfx$, rather than estimation of its sparsity. Another set of works from both the compressed sensing and signal processing literatures \cite{lopes2015compressed, reeves2010fundamental, wang2012sparsity, fletcher2009necessary} focus on the related problem of ``sparsity pattern recovery,'' which involves identifying the positions of the nonzero (or largest) entries of the vector $\bfx$, but not their values.

We will require our tests to be non-adaptive, but in contrast to existing works, we focus only on estimating the sparsity of $\bfx$, under the assumptions of absolute sparsity (as opposed to approximate sparsity) and no additional noise, over both finite fields and $\real$.

First, in \cref{sec:approximation_equivalence} we identify a necessary and sufficient condition on the matrix $M$ for $\Delta$-approximation to be possible. This condition applies to all models which we consider. In \cref{sec:group_testing}, we look at the group testing model specifically; given an upper bound $D$ on the number of defectives $d$, we demonstrate a lower bound $\Omega(\min(n,D \log n))$ on the number of measurements needed to $\Delta$-approximate $d$ without error. The lower bound uses a simple and elegant combinatorial approach to bound the size of a cover of the space of possible input vectors $\bfx$, by showing all such vectors with the same output $\bfb = M \odot \bfx$ must be elements of the same poset of the Boolean lattice. We also show that this lower bound is nearly tight by demonstrating the existence of a matrix $M$ with $m=O(\min(n, D\log n\frac{\log D}{\log \Delta}))$ rows along with a deterministic algorithm that $\Delta$-approximates $d$.

In \cref{sec:group_testing_with_output_noise} we generalize the results of the previous section by assuming adversarial output noise is added to the test results. We recover similar lower and upper bounds on the number of necessary tests in this new setting, with additional additive terms that depend on the amount of noise.
In \cref{sec:sensing_linear}, we take the operation $\odot$ to be multiplication over either $\F_q$ or $\real$, and exhibit close connections between the number of tests needed to approximate $d$ and existing quantities in coding theory; this allows us to prove lower and upper bounds on the number of necessary tests that are tight up to constant factors. Our main results are summarized in \cref{results_table}, though some are more fine-grained than the table implies.
\begin{center}
\begin{threeparttable}[!t]
\footnotesize
\caption{Number of Measurements Needed}
\label{results_table}
\begin{centering}
\begin{tabular}{|m{20mm}|c|c|c|}
\hline
 \centering \multirow{2}{*}{Model} & \multirow{2}{*}{$D=n$} &  \multicolumn{2}{c|}{$D=\mathrm{o}(n)$}  \\
 \hhline{|~|~|--}
 &       & Lower Bound                & Upper Bound \\
 \hhline{-|-|-|-}
&&&\\[-1em]
\centering Group Testing & $\Theta(n)$ & $\Omega(\frac{D}{\Delta^2} \log \frac{n}D)$ & $O(\frac{\log D}{\log \Delta} D \log n)$ \\
 \hline
\centering Compressed Sensing over $\F_q, q <  n$ & $\Theta(n)$ & $\frac{D}2 \log_q \frac{n}{D}^*$ & $2D \log_q \frac{n}{D}^*$ \\
\hline
\centering  Compressed Sensing over $\real$ or $\F_q, q \geq  n$ & $\Theta(n)$ & $D-1^{**}$ & $2D$ \\
\hline
\end{tabular}
\begin{tablenotes}
\footnotesize
\item $^*$Omitting lower order terms.
\item $^{**}$Assuming $D \geq 2 \floor{\Delta^2} - 4.$
\end{tablenotes}
\end{centering}
\end{threeparttable}
\end{center}

\section{Preliminaries and Condition for Approximability} % (fold)
\label{sec:approximation_equivalence}

Throughout, we write $\log x$ to mean $\log_2 x$, and $\normO{\bfx}$ to mean the sparsity or number of nonzero entries of the vector $\bfx$. We will denote the set $\set{1, 2, \dotsc, n}$ by $[n]$. An $\ell$-subset of a set $S$ is a simply a subset of $S$ of size $\ell$.

For the estimation problem defined in \cref{sec:intro} in all models, we  have the following necessary and sufficient condition on the matrix $M$ for it to be used to $\Delta$-approximate $d$, the sparsity of the vector in question.

{\theorem \label{estim_prob} Fix $\Delta > 1$. Let $M \in \F^{m\times n}$ be a matrix such that there exists a decoder producing a $\Delta$-approximation $\hat{d}$ of $d = \normO{\bfx}$ from observing $\bfb = M \odot \bfx$, for any $\bfx \in \F^n$, and assuming $D$ to be a known upper bound on $d$. Consider any two defective vectors $\bfv_1, \bfv_2 \in \F^n$ such that $\normO{\bfv_2} \leq \normO{\bfv_1} \leq D$. Then $M$  must have the property that
\begin{equation}\label{matrix_equiv_estim}
\frac{\normO{\bfv_1}}{\normO{\bfv_2}} > \Delta^2 \implies M \odot \bfv_1 \ne M \odot \bfv_2.
\end{equation}
Conversely, for any matrix $M$ satisfying the above property, there exists a decoder producing an estimate $\hat{d}$ that satisfies the approximation criteria in \cref{estimation_crit}.
}
\begin{proof}
	In the forward direction, given a matrix $M$ we show that if \cref{matrix_equiv_estim} is not satisfied for two vectors, then no estimation algorithm satisfying \cref{estimation_crit} can exist. If \cref{matrix_equiv_estim} is not satisfied for two vectors $\bfv_1, \bfv_2 \in \F^n$, then $\frac{\normO{\bfv_1}}{\normO{\bfv_2}} > \Delta^2$ and $M\odot \bfv_1 = M\odot \bfv_2$. Then a deterministic decoder must output the same estimate $\hat{d}$ when observing $M\odot \bfv_1$ and $M\odot \bfv_2$, but $\normO{\bfv_1}$ and $\normO{\bfv_2}$ differ by more than a $\Delta^2$ factor, so whatever $\hat{d}$ is output will violate \cref{estimation_crit} for one of $\bfv_1$ or $\bfv_2$. Thus no decoding scheme can deterministically produce an estimate $\hat{d}$ satisfying \cref{estimation_crit} in this case.

For the converse assume \cref{matrix_equiv_estim} holds for the matrix $M$. When observing the result vector $\bfb$, we define
\[
\hat{d} = \sqrt{\left( \max_{\substack{\bfy \in \F^n \\ \normO{\bfy} \leq D \\ \bfb = M \odot \bfy}} \normO{\bfy} \right) \left( \min_{\substack{\bfy \in \F^n \\ \bfb = M \odot \bfy}} \normO{\bfy} \right) },
\]
i.e., we estimate $d$ by the geometric mean of the weights of the minimum and maximum weight vectors $\bfy$ with weight $\leq D$ and $M \odot \bfy = \bfb$.
	
	From \cref{matrix_equiv_estim} we are guaranteed that the estimate above is within a $\Delta$ factor of both the lowest and highest weight vectors $\bfy$ with $\bfb = M \odot \bfy$, so will satisfy \cref{estimation_crit}.
\end{proof}

Note that the above proof of decoder existence does not imply the existence of an efficient decoder; the obvious implementation of the decoder described in the proof would take time exponential in $n$ to determine the minimum and maximum weight vectors $\bfy$ with $\bfb = M \odot \bfy$. We will see that despite this, efficient decoding is possible for some specific matrices with this property.

\section{Defective Approximation for Group Testing} % (fold)
\label{sec:group_testing}

Throughout this section we take $\F = \F_2$ and assume that the operation $\odot$ denotes logical OR as defined in \cref{sec:intro}. For a subset $S \subseteq [n]$, we write $\bfv(S)$ to mean the unique binary vector with support $S$.  We establish upper and lower bounds on the number of rows of a matrix capable of $\Delta$-approximating $d$ without error. As is standard in non-adaptive group testing, we will typically assume that an initial upper bound $D$ on the number of defectives $d$ is known. We then show that the number of rows $m$ of the matrix $M$ in \cref{measurement} must satisfy
$m = \Omega (D \log \frac{n}{D}).$
On the other hand, we show that there exists a matrix $M$ with $m = O(D \log D \log n)$ rows satisfying \cref{matrix_equiv_estim}. When no upper bound on the number of defectives is known, we can take $D=n$, and in this scenario our bound shows that a linear (in $n$) number of measurements is required to find an estimate $\hat{d}$ satisfying \cref{estimation_crit}. This implies only constant factor improvement is possible, as $n$ measurements suffice to determine $d$ exactly.

\begin{theorem} 
\label{pool_lb_D}
Suppose $M \in \F_2^{m\times n}$ is a matrix capable of $\Delta$-approximating $d$ when $d \leq D \leq \frac{n}{2}$. Then it is necessary that
\begin{align*}
m &\geq \left(\frac{D}{\Delta^2} - 1\right) \log \frac{n}{D - \Delta^2} - \left(\frac{D}{\Delta^2} -1\right) \log e \\
&= \Omega\left(\frac{D}{\Delta^2} \log \frac{n}{D}\right).
\end{align*}
\end{theorem}

\begin{proof}
	Assume two sets $S_1, S_2 \subseteq [n]$ satisfy $M \odot \bfv(S_1) = M \odot \bfv(S_2)$.
	Then using the definition of the operation $\odot$ as logical OR, we have
	$M \odot \bfv(S_1) = M \odot \bfv(S_2) = M \odot \bfv(S_1 \cup S_2).$

	Let $\bb{P}_D([n])$ denote the set $\set{A \subseteq [n]:\abs{A}\leq D}$ of possible defective subsets with size at most $D$. Let $\cP_1, \cP_2, \ldots, \cP_t$ denote a partition of $\bb{P}_D([n])$ such that $A, B\in \cP_i$ if and only if $M \odot \bfv(A) = M \odot \bfv(B).$ Then we must have
	\begin{equation}
	\label{poset_condition1}
		A, B \in \cP_i \implies A\cup B \in \cP_i.
	\end{equation}
	Since the matrix $M$ is capable of $\Delta$-approximation, by \cref{estim_prob} any two sets $A, B$ with $\abs{B}\leq \abs{A}\leq D$ in the same part of the partition $\cP_i$ must also satisfy
	\begin{equation}\label{poset_condition2}
		\abs{A} \leq \Delta^2 \cdot \abs{B}.
	\end{equation}
	
	% Consider the lattice $L$ on $2^{[n]}$ ordered by inclusion. Then the $\cP_i$ must form posets in $L$ with unique maximal elements, as they are closed under union (the maximal element may have size more than $D$). 

	For any matrix $M$, let $t$ be the number of distinct vectors $\bfb$ arising as $\bfb = M \odot \bfx$ for any potential vector of defectives $\bfx$. Then each row of $M$ corresponds to a test with either a positive or negative result, so there can be at most $2^m$ distinct vectors $\bfb$. Thus $t \leq 2^m$, so $m \geq \log t$. Since any partition is also a cover, we lower bound $t$ by finding a lower bound on the minimum size of a cover of $\bb{P}_D([n])$ with posets (ordered by inclusion) $P_1, P_2, \ldots, P_t$ in $\bb{P}_D([n])$ satisfying \cref{poset_condition1,poset_condition2}.

	From \cref{poset_condition2}, we see that any poset $P$ containing an $\ell$-subset cannot contain elements of size between $\ell \Delta^2$ and $D$, inclusive. Assume that $P$ contains an $\ell$-subset for $\ell < \frac{D}{\Delta^2}$, so $\ell \leq \frac{D}{\Delta^2} - 1$. Suppose that $P$ has more than one maximal element in $\bb{P}_D([n])$, and call two such elements $A$ and $B$. Since all elements of $P$ have size $\leq \ell \Delta^2 \leq D - \Delta^2$, the union of $A$ with any one element of $B \setminus A$ must be of size at most $D - \Delta^2 + 1 \leq D$. But since the union must be in the poset as well, this contradicts the maximality of $A$ and $B$. Thus for $\ell < \frac{D}{\Delta^2}$ there exists a unique maximal element in any poset containing an $\ell$-subset.

	Then as a poset $P$ containing an $\ell$-subset of $[n]$ cannot have an element of size greater than $\ell \Delta^2$, we know that the minimum size $t_\ell$ of a cover of all $\ell$-subsets of $[n]$ with posets satisfying \cref{poset_condition1,poset_condition2} is at least $t_\ell \geq {\binom{n}{\ell}} / {\binom{\ell \Delta^2}{\ell}}.$

	The minimum size of any cover of $\bb{P}_D([n])$ is at least the size of the minimum cover of the $\ell$-subsets of $[n]$, for any particular value of $\ell \leq D$. Thus, $t  \geq t_{\frac{D}{\Delta^2} - 1}$. Then as $m \geq \log t$, we have
	\begin{align*}
	m & \geq \Log{ \binom{n}{\frac{D}{\Delta^2} - 1} / \binom{D - \Delta^2}{\frac{D}{\Delta^2} - 1} } \\
			    & \geq \Log{ \br*{ \frac{n}{D/\Delta^2 - 1}}^{\frac{D}{\Delta^2} - 1} / \left(\frac{(D - \Delta^2)^{\frac{D}{\Delta^2}  -1}}{\br{\frac{D}{\Delta^2} - 1}!}\right) } \\
			   & = \left(\frac{D}{\Delta^2} - 1\right) \log n - \left(\frac{D}{\Delta^2} - 1\right) \log \left(\frac{D}{\Delta^2} -1\right) - \left(\frac{D}{\Delta^2} - 1\right) \log (D - \Delta^2) + \log \left(\left(\frac{D}{\Delta^2} - 1\right)!\right) \\
			   & \geq \left(\frac{D}{\Delta^2} - 1\right) \log \frac{n}{D - \Delta^2} - \left(\frac{D}{\Delta^2} -1\right) \log e,
	\end{align*}
	using Stirling's approximation in the last step.
\end{proof}

Now consider the case when no nontrivial upper bound on $d$ is known. When $D = n$ we have the following result.

\begin{corollary}
Fix $\Delta > 1$. Consider a matrix $M \in \F_2^{m\times n}$ that can be used for $\Delta$-approximation in the group testing model with only the trivial upper bound $D = n$ on the true number of defectives $d$. Then for large $n$, it is necessary that
\begin{equation*}
	m \geq n \left( h \left( \frac{1}{\Delta^4} \right) - \frac{1}{\Delta^2} h \left( \frac{1}{\Delta^2} \right) \right) = \Omega(n),
\end{equation*}
where $h(x)$ denotes the binary entropy function.
\end{corollary}

\begin{proof}
From the proof of the previous theorem, we have
\begin{align*}
	m &\geq \max_{l\leq n} \log \left( \binom{n}{l} / \binom{l\Delta^2}{l} \right) 
	  \geq \log \left( \binom{n}{n/\Delta^4} / \binom{n/\Delta^2}{n/\Delta^4} \right)
\end{align*}
by setting $l = \frac{n}{\Delta^4}$. For large $n$ using Stirling's approximation, we have that the last term above converges to 
$
n \left( h \left( \frac{1}{\Delta^4} \right) - \frac{1}{\Delta^2} h \left( \frac{1}{\Delta^2} \right) \right). 
$
This expression is positive for $\Delta > 1$, so the entire expression is $\Omega(n)$ assuming $\Delta$ is a fixed constant.
\end{proof}

We now show the existence of a matrix $M$ capable of $\Delta$-approximation with $m=O(\frac{D \log D}{\log \Delta} \log n)$ rows for any $D$ and $\Delta>1$. We modify a construction of Damaschke and Muhammad \cite{damaschke2010competitive} for non-adaptive sparsity approximation with error, and show that with additional repetition of certain tests, we can achieve no error.

\begin{theorem}
\label{noiseless_UB}
Given an initial upper bound $D$ on the number of defectives, there exists a matrix $M \in \F_2^{m \times n}$ that can $\Delta$-approximate $d$ in the group testing model such that
$m = O(\frac{D \log D}{\log \Delta} \log n).$ Furthermore, this matrix has a decoding scheme that requires only $O(m) =  O(\frac{D \log D}{\log \Delta} \log n)$ time.
\end{theorem}
\begin{proof}

	Let $\be(p)$ denote the Bernoulli variable on $\F_2$ such that $\prob{\be(p)=1}=p$. Let $\hat{d}$ denote the estimate of $d$, which is specific to the matrix $M$. We show that a matrix constructed in a random way combined with a specific estimator $\hat{d}$ works simultaneously for all vectors of defectives with nonzero probability. To this end, we define the bad events $E_1$ and $E_2$ to be the events that our estimate is too far off, the complements of the estimation criteria in \cref{estimation_crit}:
	\begin{subequations}\label{bad_events}
	\begin{align}		
		&E_1 : \frac{\hat{d}}{d} > \Delta \label{bad_E1}, \quad \\
		&E_2 : \frac{\hat{d}}{d} < \frac{1}{\Delta} \label{bad_E2}.
	\end{align}
	\end{subequations}

	\noindent{\textbf{Construction:}} We modify the random construction in \cite{damaschke2010competitive} to construct $M \in \F_2^{m \times n}$. For this matrix $M$, we show that the probability of either bad event occurring over all defective vectors with $d\leq D$ is strictly less than $1$ when $m = O(\frac{D \log D}{\log \Delta} \log n)$. Thus there exists a matrix $M$ with $m$ rows for which none of the bad events in \cref{bad_events} occur for any defective vector of weight at most $D$.

Consider a fixed parameter $b>1$, and let $\delta := \log_b \Delta$. Take $s$ to be a fixed parameter such that $s < \delta -1$, and let $l \in \{\floor{\log_b \ln 2}, \floor{\log_b \ln 2} + 1, \dotsc, \ceil*{\log_b{D}}\}$ denote indices for subsets of tests. For each index $l$, we construct $t$ random identically and independently  distributed (iid) tests such that each element is selected in the test for index $l$ with probability $1-(1-\frac{1}{D})^{b^l}$. Then the total number of such indices is $N := \ceil*{\log_b{D}}-\floor{\log_b \ln 2}+1$, so the matrix $M$ consists of $t N$ rows with the elements in row indices $j\in \set{(l-1)t+1, (l-1)t+2, \dotsc, lt}$ selected randomly and independently as $\be(1-(1-\frac{1}{D})^{b^l})$. Thus the probability of row $j$ for $j\in \set{(l-1)t+1, (l-1)t+2, \dotsc, lt}$ having a negative result (containing no defectives) is $q_l(d) := (1-\frac{1}{D})^{d b^l}$. Now, there exists an index $\ell(d) \in \set{\floor{\log_b \ln 2}, \floor{\log_b \ln 2} + 1, \dotsc, \ceil*{\log_b{D}}}$ (or simply $\ell$) such that
	$
		\frac{1}{2} \leq q_{\ell-s}(d) < \frac{1}{2^{1/b}},
	$
	because $q_{l+1} = q_l^b$.

	Let $L \in \set{\floor{\log_b \ln 2}, \floor{\log_b \ln 2} + 1, \dotsc, \ceil*{\log_b{D}}}$ denote the random variable corresponding to the maximum index for which the majority of test results were negative. Our decoding algorithm will be to take the estimate $\hat{d}$ of $d$ such that $q_{L-s}(\hat{d})=\frac{1}{2}$, i.e.,
	$
		\hat{d} = \frac{-1}{b^{L-s} \log(1-\frac{1}{D})}.
	$
	Then the probability of bad event $E_1$ for a defective set of size $d$ is
	\begin{align*}
		\prob{E_1} &= \prob{\frac{\hat{d}}{d} > \Delta} \\
				   &\leq \prob{\hat{d} \geq d \Delta} \\
				   &= \prob{\frac{1}{2} \leq \br*{1-\frac{1}{D}}^{d \;b^{L+\delta-s}}} \\
				   &= \prob{\frac{1}{2} \leq q_{L+\delta-s}(d)}\\
				   &= \prob{L \leq \ell(d)-\delta } \\
				   &= \prod_{\ell(d)-\delta<l\leq \ceil*{\log_b{D}}} \hspace{-1.5mm} F \left(\ceil*{\frac{t-1}{2}};t,q_l(d)\right),
	\end{align*}
	where $F(k;n,p)=\prob{X\leq k}$ denotes the cumulative distribution function for $X = \sum_{i=1}^n X_i$ for $X_i \sim \be(p)$ iid. Similarly, for bad event $E_2$ we have
	\begin{align*}
		\prob{E_2} &= \prob{\frac{\hat{d}}{d} < \frac{1}{\Delta}} \\ 
				   &\leq \prob{\hat{d} \leq d /\Delta} \\
				   &= \prob{\frac{1}{2} \geq \br*{1-\frac{1}{D}}^{d \;b^{L-\delta-s}}} \\
				   &= \prob{\frac{1}{2} \geq q_{L-\delta-s}(d)}\\
				   &= \prob{L \geq \ell(d)+\delta } \\
				   &= \sum_{l= \ell(d)+\delta }^{\ceil*{\log_b{D}}} F \left(\ceil*{\frac{t-1}{2}};t,1-q_l(d)\right).
	\end{align*}

	Thus the probability of the event $\tilde{E_2},$ defined to be the union of the events $E_2$ for all defective sets of size $d\leq D$, can be upper bounded by union bound as follows:
	\begin{align*}
		\prob{\tilde{E_2}} &\leq \sum_{1\leq i\leq D} \binom{n}{i} \prob{E_2}  \\
						   &\leq \sum_{i} \binom{n}{i} \sum_{\ell(d)+\delta \leq l\leq \ceil*{\log_b{D}}} F\left(\ceil*{\frac{t-1}{2}};t,1-q_l(i)\right)  \\
						   &\leq \sum_{i} \binom{n}{i} \cdot N \cdot F \left(\ceil*{\frac{t-1}{2}};t,1-q_{\ceil*{\log_b{D}}}(i)\right)  \\
						   &\leq \sum_{i} \binom{n}{i} \cdot N \cdot \Exp{-t \;D\left(\frac{1}{2}||1-q_{\ceil*{\log_b{D}}}(i)\right)},
	\end{align*}
	where $D(a||p)\defeq a \log(\frac{a}{p})+ (1-a) \log(\frac{1-a}{1-p})$ denotes the KL divergence between $a$ and $p$, and the last inequality follows from the bound on $F(k;n,p)$ in \cite{Arratia1989}. Then applying the inequalities $\binom{n}{i} \leq n^i$ and $q_{\ceil{\log_b D}} \leq e^{-i}$, we have that $\prob{\tilde{E_2}}$ is at most
	\begin{align*}
		N \sum_{1\leq i\leq D} \Exp{i \ln n +\frac{t}{2} \left(1 + \log(1-q_{\ceil{\log_b D}}) - \frac{i}{\ln 2}\right)}.
	\end{align*}
	Now, it can be seen that for $t=O(\log n)$, $\prob{\tilde{E_2}}$ goes to $0$.

	Similarly, we bound the probability of $\tilde{E_1}$,  the union of the events $E_1$ for all defective sets of size $d\leq D$, recalling that $\delta>s+1$:
	\begin{align*}
		\prob{\tilde{E_1}} &\leq \sum_{1\leq i\leq D} \binom{n}{i} \prob{E_1}  \\
						   &\leq \sum_{i} \binom{n}{i} \prod_{\ell(i)-\delta < l\leq \ceil*{\log_b{D}}} F \left(\ceil*{\frac{t-1}{2}};t,q_l(i)\right)  \\
						   &\leq \sum_{i} \binom{n}{i} \prod_{\ell(i)-\delta < l\leq \ell(i)-s} F \left(\ceil*{\frac{t-1}{2}};t,q_{l}(i)\right)  \\
						   &\leq \sum_{i} \binom{n}{i} \Exp{-\frac{\delta - s - 1}{2 q_{\ell(i)-s-1}} \cdot t\left(q_{\ell(d)-s-1}-\frac{1}{2}\right)^2},
	\end{align*}
	where the last inequality follows from the Chernoff bound applied to the binomial distribution. Thus $\prob{\tilde{E_1}}$ is at most
	\begin{align*}
		&D \cdot \Exp{D \ln n- \frac{\delta - s - 1}{2 q_{\ell(d)-s-1}} \cdot t \left(q_{\ell(d)-s-1}-\frac{1}{2}\right)^2} \\
						   \leq \  &D \cdot \Exp{D \ln n- \frac{\delta - s - 1}{2^{1-b^{-1}}} \cdot t \left(\frac{1}{2^{b^{-1}}}-\frac{1}{2}\right)^2}.
	\end{align*}
	It can be seen that for $t=O(\frac{D}{\delta} \log n)$, $\prob{\tilde{E_1}}$ goes to $0$.

	Therefore when the number of tests $m = O(\frac{D \log D}{\log \Delta}  \log n)$, there exists a test matrix and estimation algorithm which estimates $d$ within a multiplicative factor of $\Delta$ for all defective vectors of weight $\leq D$.

The decoding requires only computing a function of the largest index $l$ for which the corresponding block of $t$ tests had a majority of test results negative, and this index can easily be determined in a single pass over the result vector, requiring $O(m)$ time.
\end{proof}

\section{Group Testing Approximation with Output Noise} % (fold)
\label{sec:group_testing_with_output_noise}
	We now consider the group testing scenario with the additional complication that the output $\bfb = M\odot \bfx$ is corrupted by noise. Assume that the output vector $\bfy$ contains at most $e_0$ false positives (0s flipped to 1s) and at most $e_1$ false negatives (1s flipped to 0s). We will aim to bound the number of tests required to $\Delta$-approximate the number of defectives in such a noisy setting. 

	Denote the support set of a vector $\bfx \in \F_2^m$ as $\supp{\bfx} \subseteq [n]$. We first give a necessary and sufficient condition under which a non-adaptive group testing matrix and deterministic estimator capable of $\Delta$-approximation exist in the presence of bounded asymmetric noise.

	{\definition For  $\bfx, \bfy \in \F_2^m$, $(\bfx, \bfy)$ are called $(e_0,e_1)$-far iff $\abs{\supp{\bfy} - \supp{\bfx}} > e_0 $ or $\abs{\supp{\bfx} - \supp{\bfy}} > e_1$. If neither condition occurs, then $(\bfx, \bfy)$ are said to be $(e_0, e_1)$-close.}
\vspace{2mm}

We note that closeness is in general not symmetric; if $(\bfx, \bfy)$ are $(e_0, e_1)$-close, then $(\bfy, \bfx)$ are $(e_1, e_0)$-close, but need not be $(e_0, e_1)$-close.

	{\lemma \label{asymmetric_balls} Consider two vectors $\bfx_1, \bfx_2 \in \F_2^{n}$ such that $(\bfx_1,\bfx_2)$ are $(e_0+e_1,e_0+e_1)$-close. Then there exists a vector $\bfy \in \F_2^n$ such that both $(\bfx_1,\bfy)$ and $(\bfx_2,\bfy)$ are $(e_0,e_1)$-close.}
	\begin{proof}
		Let $X_1 \defeq \supp{\bfx_1},$ $X_2 \defeq \supp{\bfx_2}$. We define the support of the vector $\bfy$, $Y \defeq \supp{\bfy}$ as
		$$ Y \defeq (X_1 \cap X_2) \cup P \cup R, $$
		where $P \subseteq X_2 \setminus X_1$ such that $\abs{P} = \min\br{\abs{X_2 \setminus X_1}, e_0}$ and $R \subseteq X_1 \setminus X_2$ such that $\abs{R} = \min\br{\abs{X_1 \setminus X_2}, e_0}.$

		Now $\supp{Y \setminus X_1} = P$ and $\supp{Y \setminus X_2} = R$. Thus, $\abs{Y \setminus X_1} = \abs{P} \leq e_0$ and $\abs{X_1 \setminus Y} = \abs{X_1\setminus X_2}-\abs{R} \leq e_1$. Similarly, $\abs{Y \setminus X_2} = \abs{R} \leq e_0$ and $\abs{X_2\setminus Y} = \abs{X_2 \setminus X_1}-\abs{P} \leq e_1$.
	\end{proof}
	
	We can now state the condition for a matrix to be capable of $\Delta$-approximation in this noisy setting.

	{\theorem Let $\Delta > 1$. Let $M \in \F^{m\times n}$ be a matrix such that there exists a decoder producing a $\Delta$-approximation $\hat{d}$ of $d = \normO{\bfx}$ for any $\bfx \in \mathbb{F}^n$ of weight at most $D$ when observing 
		$$\bfy = M \odot \bfx+ \bfn,$$
		where $\bfn$ denotes a noise vector such that $(M \odot \bfx,\bfy)$ are $(e_0,e_1)$-close. Then for any two sets $S_1, S_2 \subseteq [n]$ with $|S_2| \leq |S_1| \leq D$, $M$ must satisfy
		\begin{align}\label{matrix_equiv_estim_noisy}
		\frac{\abs{S_1}}{\abs{S_2}} > \Delta^2 \implies (M \odot \bfv(S_1), M \odot \bfv(S_2)) \nonumber \\ \mbox{ are $(e_0+e_1, e_0+e_1)-$far. }
		\end{align}
		Conversely, for any matrix $M$ satisfying \cref{matrix_equiv_estim_noisy}, there exists a decoder producing an estimate $\hat{d}$ that satisfies the approximation criteria in \cref{estimation_crit} for all $\bfx$ of weight at most $D$.}
	
	\begin{proof}
	Consider a matrix $M \in \F_2^{m\times n}$ that does not satisfy equation \cref{matrix_equiv_estim_noisy}. Then there must exist sets $S_1, S_2 \subseteq [n]$ such that $\abs{S_1}/\abs{S_2} > \Delta^2$ and $(M\odot \bfv(S_1),M\odot \bfv(S_2))$ are $(e_0+e_1, e_0+e_1)$-close. Thus, there exists a vector $\bfy \in \F_2^m$ such that $(M\odot \bfv(S_1),\bfy)$ and $(M\odot \bfv(S_2),\bfy)$ are both $(e_0,e_1)$-close from \cref{asymmetric_balls}. As a result, there exist valid noise vectors $\bfn_1, \bfn_2$ with the property that $\bfy = M \odot \bfx_1 + \bfn_1 = M \odot \bfx_2 + \bfn_2$. Since the estimate $\hat{d}$ must be the same for both cases, the estimate cannot satisfy the approximation criteria in both cases simultaneously.

	For sufficiency, consider the following set $\cS = \set{S \subset [n] : (M\odot \bfv(S),\bfy) \mbox{ are $(e_0,e_1)-$close}}$. We take 
	$$\hat{d} = \sqrt{\left(\min_{S\in \cS} \abs{S}\right)  \left(\max_{S\in \cS} \abs{S}\right)}$$

	Our estimate $\hat{d}$ is then within a $\Delta$ factor of the cardinality of any set in $\mathcal{S}$, so as the actual defective set must belong to $\cS$, $\hat{d}$ satisfies the approximation criteria.
	\end{proof}
	
	For an upper bound on the minimum size of such a matrix, we can make a simple modification to the construction used in the noiseless case. Recall that there, we constructed a matrix with $l = O(\log D)$ ``indices,'' where each index consisted of $t = O(D \log n)$ tests, with each element selected for the test uniformly at random with probability $p_l$. Then to decode, we computed a function of the largest index $l$ for which the corresponding block of $t$ tests had the majority of test results negative. Adding a noise vector to the output flips the results of at most $e_0 + e_1$ tests, so to ensure this construction is resilient to noise, we simply add $2(e_0+e_1)$ tests to each index $l$, again with each element chosen uniformly to be tested with probability $p_l$. Then we are guaranteed that the majority of test results for each index will be the same as it would have been in the noiseless case, so the argument of \cref{noiseless_UB} applies to show correctness. Thus we have the following theorem.
	
	\begin{theorem}
	\label{noisy_UB}
	There exist $m \times n$ matrices capable of $\Delta$-approximation of all vectors with at most $D$ defectives, when the output is corrupted by noise so that at most $e_0$ negative results become positive, and at most $e_1$ positive results become negative, with
	\begin{equation*}
	m = O\left(\frac{D \log D}{\log \Delta} \log n + (e_0 + e_1) \log D\right)
	\end{equation*}
	rows. Furthermore, some such matrices have decoding schemes that require only $O(m)$ time.
	\end{theorem}

	Next we turn to a lower bound on the number of rows of such a matrix. Consider the undirected graph $G = (V,\cE)$ on the vertex set $V = \set{A \subseteq [n] : \abs{A}\leq D}$. The edges in $G$ are defined as follows: $(A,B) \in \cE$ if and only if $(M\odot \bfv(A),M\odot \bfv(B))$ are $(e_0+e_1,e_0+e_1)$-close. This is a ``confusion graph'' for the set of possible inputs; by \cref{asymmetric_balls} there is an edge between vertices corresponding to two inputs whenever there exists a vector that both the corresponding outputs could map to after being corrupted by noise, so these inputs are ``confusable.'' Then an independent set $\cI$ in the graph $G$ corresponds to a set of inputs for which the corresponding outputs are all pairwise $(e_0+e_1,e_0+e_1)$-far. 

		{\lemma \label{prop_graph} For a matrix $M$ satisfying the noisy approximation criteria in \cref{matrix_equiv_estim_noisy}, the confusion graph $G$ as defined above must satisfy the following conditions, for $A,B \in V$:
		\begin{itemize}
			\item[i)] If $(A, B) \in \cE$ and $|A \cup B| \leq D$, then $(A\cup B,A) \in \cE$ and $(A\cup B,B) \in \cE$.
			\item[ii)] If $\frac{\abs{A}}{\abs{B}} \not\in \brac*{\frac{1}{\Delta^2},\Delta^2}$, then $(A,B) \not\in \cE$.
	 	\end{itemize}}
 	\begin{proof} 		
		For the first condition, by definition of the edge set we have $(A, B)$ are $(e_0 + e_1, e_0 + e_1)$-close. Then (as long as $A \cup B$ corresponds to a vertex in $G$, guaranteed by $|A \cup B| \leq D$), $(A \cup B, A)$ are $(0, e_0 + e_1)$-close, so have an edge between them, and similarly for $(A \cup B, B)$. The second condition follows immediately from the fact that $M$ satisfies \cref{matrix_equiv_estim_noisy}.
 	\end{proof}
	Thus, any independent set in $G$ corresponds to a packing of $\F_2^m$ with balls of (asymmetric) radius $(\frac{e_0+e_1}{2},\frac{e_0+e_1}{2})$, defined as follows:
	$$B_{as}(\bfx,(r_1, r_2)) \defeq \set{\bfy \in \F_2^m : \abs{\supp{\bfx}-\supp{\bfy}}\leq r_1 \mbox{ and } \abs{\supp{\bfy}-\supp{\bfx}}\leq r_2 }.$$

	The size of the largest independent set, $\cI_{\max}$, can be upper bounded as
	\begin{equation}
		\abs{\cI_{\max}} \leq A\left(m,\left(\frac{e_0+e_1}{2},\frac{e_0+e_1}{2}\right)\right),
	\end{equation}
	where $A(m,(\frac{e_0+e_1}{2},\frac{e_0+e_1}{2}))$ denotes the size of the maximum packing of $\F_2^m$ with balls of radius $(e_0+e_1,e_0+e_1)$. 

	Then we have
	\begin{align}
		\min_{\bfx\in \F_2^m}\abs{B_{as}(\bfx,(r, r))} &= \min_{\bfx\in \F_2^m}\sum_{i=0}^{\Min{r,\abs{\bfx}}} \sum_{j=0}^{\Min{r,n-\abs{\bfx}}}  \binom{\abs{x}}{i} \binom{n-\abs{x}}{j} \nonumber \\
						&\geq \min_{\bfx\in \F_2^m}  \binom{\abs{x}}{\Min{r,\frac{\abs{\bfx}}{2}}} \binom{m-\abs{x}}{\Min{r,\frac{m-\abs{\bfx}}{2}}}\nonumber \\
						&\geq \min_{\substack{\bfx\in \F_2^m \\ 2r \leq |x| \leq m - 2r}} \binom{\abs{x}}{r} \binom{m-\abs{x}}{r}\nonumber \\
						&\geq \min_{\substack{\bfx\in \F_2^m \\ 2r \leq |x| \leq m - 2r}} \left(\frac{|x|}{r}\right)^r \left(\frac{m - |x|}{r}\right)^r \nonumber \\
						&\geq \frac{(2r)^r (m-2r)^r}{r^{2r}}. \nonumber
	\end{align}
	We use the sphere-packing bound to upper bound $A(m,(e_0+e_1,e_0+e_1))$:
	\begin{align}
		A(m,(e_0+e_1,e_0+e_1)) &\leq \frac{2^m}{\min_{\bfx \in \F_2^m} \abs{B_{as}(\bfx,(\frac{e_0+e_1}{2},\frac{e_0+e_1}{2}))}} \nonumber \\
							   &\leq (2^m) / \left(\frac{(e_0+e_1)^{(e_0+e_1)/2} \cdot (m-(e_0+e_1))^{(e_0+e_1)/2}}{(\frac{e_0+e_1}{2})^{e_0+e_1}}\right).
	\end{align}

	Note that since the graph $G$ is an edge super-graph of the confusion graph corresponding to the noiseless case, the lower bound on the number of partitions in \cref{pool_lb_D} can also be used to lower bound the size of an independent set in $G$. Thus, we have the following lower bound on $m$.
	
	\begin{theorem}
	\label{noisy_LB}
	Let $E$ be the maximum weight of a noise vector, and assume $E$ grows no faster than $O(\frac{D}{\Delta^2} \log n)$. For an $m \times n$ matrix $M$ capable of $\Delta$-approximation when the output is corrupted by such a noise vector, assuming the weight of the input vector of defectives is at most $D \leq \frac{n}{2}$, we must have
	\begin{align*}
	m &\geq \left(\frac{D}{\Delta^2} - 1\right) \log \frac{n}{D - \Delta^2} - \left(\frac{D}{\Delta^2} -1\right) \log e + E - \frac{E}{2} \log E +\Omega\left(\frac{E}{2} \log \left(\frac{D}{\Delta^2} \log n - E\right)\right) \\
	&= \Omega\left(\frac{D}{\Delta^2} \log n\right) + \Omega\left(E \log \left(\frac{D}{\Delta^2} \log n - E\right)\right).
	\end{align*}
	\end{theorem}
	\begin{proof}
	Let $e_0$ be the number of 0s flipped to 1s, and $e_1$ the number of 1s flipped to 0s, so $E = e_0 + e_1$. From the above discussion, by plugging in the bound from \cref{pool_lb_D}, we have
	\begin{equation*}
	\left(\frac{D}{\Delta^2} - 1\right) \log \frac{n}{D - \Delta^2} - \left(\frac{D}{\Delta^2} -1\right) \log e \leq \log 2^m - \log \frac{E^{E/2} \cdot (m-E)^{E/2}}{(\frac{E}{2})^{E}},
	\end{equation*}
	thus
	\begin{eqnarray*}
	m &\geq& \left(\frac{D}{\Delta^2} - 1\right) \log \frac{n}{D - \Delta^2} - \left(\frac{D}{\Delta^2} -1\right) \log e + \log \frac{E^{E/2} \cdot (m-E)^{E/2}}{\left(\frac{E}{2}\right)^{E}}\\
	&\geq& \left(\frac{D}{\Delta^2} - 1\right) \log \frac{n}{D - \Delta^2} - \left(\frac{D}{\Delta^2} -1\right) \log e + \frac{E}{2} \log (E) + \frac{E}{2} \log (m-E) - E \log \left(\frac{E}{2}\right) \\
	&=& \left(\frac{D}{\Delta^2} - 1\right) \log \frac{n}{D - \Delta^2} - \left(\frac{D}{\Delta^2} -1\right) \log e + E - \frac{E}{2} \log E +\frac{E}{2} \log (m-E).
	\end{eqnarray*}
	As $m$ is $\Omega(\frac{D}{\Delta^2} \log n)$ by the noiseless lower bound, we have at least
	\begin{equation*}
	\frac{E}{2} \log (m-E) = \Omega\left(\frac{E}{2} \log \left(\frac{D}{\Delta^2} \log n - E\right)\right),
	\end{equation*}
	from which the result follows.
	\end{proof}

% section group_testing_with_output_noise (end)

\section{Defective Approximation by Linear Operations} % (fold)
\label{sec:sensing_linear}

In this section we assume that the operation $\odot$ denotes standard matrix multiplication over a field $\F$, and as such will typically write $M \bfx$ instead of $M \odot \bfx$. We take $\F$ to be either finite or $\real$, the latter of which is the setting of the well-studied compressed sensing problem. Our aim now is to bound the size of a matrix $M$ that satisfies the criteria for $\Delta$-approximation in this model. Consider a vector space $V \subseteq \F^n$. Call such a vector space $(\Delta, D)$-distinguishing if it has the property that for parameters $\Delta > 1$, $D \leq n$, any two vectors $\bfx$ and $\bfy$ in the same coset of $V$ (so there exists $\bfv_1, \bfv_2 \in V$ such that $\bfx = \bfz + \bfv_1, \bfy = \bfz + \bfv_2$ for some $\bfz \in \F^n$) both having weight at most $D$ differ in weight by a factor of at most $\Delta^2$. In other words, if $\normO{\bfx} \leq \normO{\bfy}$, then $\normO{\bfy} \leq \Delta^2 \normO{\bfx}$. We can use this property to lower bound the rank (and thus the number of rows) of any $\Delta$-approximating matrix.

{\theorem \label{thm:equiv_rankM} For an $m \times n$ matrix $M$ which can $\Delta$-approximate the sparsity, $d$, of any vector with sparsity at most $D$ in the linear operations model, it is necessary that
\begin{equation}\label{equiv_rankM}
 	\rank M \geq n - \max_{(\Delta, D) \textrm{-distinguishing } V} \dim(V).
 \end{equation}}
\begin{proof}
	We show that the nullspace of any such matrix $M$ is a $(\Delta, D)$-distinguishing vector space. By our necessary criteria for $\Delta$-estimation, we have for any vectors $\bfx, \bfy$ with  $\normO{\bfx}\leq \normO{\bfy} \leq D$, that when $M \bfx = M \bfy$, then $\normO{\bfy} \leq \Delta^2 \normO{\bfx}.$
	Now, since $M (\bfx - \bfy) = 0$ if and only if $\bfx$ and $\bfy$ are in the same coset of $\ker{M}$, we have $\rank{M} \geq \min \{\rank N: \ker{N} \textrm{ is } (\Delta, D) \textrm{-distinguishing}\}$.
\end{proof}

We use the condition in \cref{thm:equiv_rankM} to bound the number of measurements needed, as $\rank{M}$ is a lower bound on the number of rows of $M$. Let $A_\F(n,d_{\min})$ denote the maximum dimension of a subspace in $\F^n$ that does not contain any nonzero vector in $\set{\bfx \in \F^n : \normO{\bfx} < d_{\min}}$. Note that for finite fields $\F_q$, the quantity $A_{\F_q}(n,d_{\min})$ denotes the maximum dimension of a linear $q$-ary code with minimum distance $d_{\min}$. A generalization of this quantity is studied in \cite{AbbeAB14}, which defines $A_\F(n,a,b)$ as the maximum dimension of a subspace in $\F^n$ that does not contain any vector in the annulus $\set{\bfx \in \F^n : a < \normO{\bfx} < b}$. Let $U$ be a $(\Delta, D)$-distinguishing vector space, and let $\bfx$ be a nonzero vector of maximum weight in $U$ subject to $\normO{\bfx}< D$. If no such vector exists, then
$\Dim{U} \leq A_\F(n,D).$
Otherwise, there exists a vector $\bfy \in \F^n$ with $\normO{\bfy}=1$ and $\normO{\bfy+\bfx} = \normO{\bfx}+1  \leq D$. Since $\bfy = \mathbf{0} + \bfy$ and $\bfx+\bfy$ are in the same coset of $U$, we have
\begin{align*}
	\normO{\bfx+\bfy} \leq \Delta^2 \underbrace{\normO{\bfy}}_{=1}
	\implies \normO{\bfx} \leq \floor{\Delta^2}-1.
\end{align*}
Thus,
$
	\dim{U} \leq A_\F(n,\floor{\Delta^2}-1,D).
$
As $A_\F(n,D)\leq A_\F(n,a,D)$ for any $a\leq D$, the above inequality applies to any $(\Delta, D)$-distinguishing space $U$. Therefore we have
\begin{equation*}
	\rank{M} \geq {n-A_\F(n,\floor{\Delta^2}-1,D)}
\end{equation*}
for any matrix $M$ that can $\Delta$-approximate $d$ for every defective vector of sparsity up to $D$ in the linear operations model.
Define $m^\star_\F(a,b,n) \defeq n - A_\F(n,a,b)$. We have from \cite[Proposition 1]{AbbeAB14} that whenever $2a - 2 \leq b$,
$
	m^\star_\F(a,b,n) = m^\star_\F(1,b,n-a+1).
$
The work in \cite{AbbeAB14} is oriented towards finite fields, but their proof is independent of the choice of field $\F$, so in particular applies also when $\F = \real$. Note that $m^\star_\F(1,b,n)$ is just the rank of the parity check matrix of the largest dimension linear code on $n$ coordinates with minimum distance $b$. Thus, as long as $D \geq 2 \floor{\Delta^2} - 4$,
\begin{equation*}
	\rank{M} \geq m^\star_\F(1,D,n-\floor{\Delta^2}+2).
\end{equation*}

If $\F$ is a finite field of size greater than or equal to $n$, or $\F = \real$, we know that
$
m^\star_\F(1,D,n-\floor{\Delta^2}+2) \geq D - 1
$
by the Singleton bound. Since we can exactly identify the set of defectives for $d \leq D$ in $2 D$ queries using basic techniques from compressed sensing, there is at most a factor 2 improvement possible.

For $\abs{\F}=q < n$, a stronger lower bound is possible since in general the Singleton bound is not tight. From the sphere-packing bound on $A_{\F}(n-\floor{\Delta^2}+2,1,D)$, we have that $m^\star_{\F_q}(1,D,n-\floor{\Delta^2}+2)$ is lower bounded by
\begin{align*}
	\left\lfloor\frac{D-1}{2}\right\rfloor \log_q(n - \floor{\Delta^2}+2) &-O(D \log D) \\
	&= \Omega\left(D \log \frac{n}{D}\right),
\end{align*}
as long as $\floor{\frac{D-1}{2}} \leq \frac{n}{2}$.

For an upper bound on the minimum possible rank of $M$ when $\abs{\F}=q < n$, recall that such a matrix must have the property that in every coset of the nullspace, any two vectors of weight less than or equal to $D$ must differ in weight by at most a $\Delta^2$ factor. As the difference of any two vectors in the same coset lies in the nullspace, this condition is satisfied if every nonzero vector in the nullspace has weight at least $2D + 1$, so we have the following result.
\begin{theorem}
\label{thm:linear_ub}
In the linear operations model, when $|\F| = q < n$, and $D \leq \frac{n(q-1) - q}{2}$, there exist matrices $M$ that can $\Delta$-approximate the true number of defectives, $d$, for every possible vector of defectives of weight at most $D$, with at most $n H_q(\frac{2D+1}{n})$ rows, where $H_q$ is the q-ary entropy function, $H_q(x) := x \log_q(q-1) - x \log_q(x) - (1-x) \log_q(1-x)$.
\end{theorem}
\begin{proof}
We have
\begin{align*}
	\min_{\substack{M \\ M \  \Delta-\textrm{approximates } d}} \rank{M}  &\leq \min_{\substack{M \\ \forall \bfx \neq 0 \in \ker{M}, \ \normO{\bfx}>2D}} \rank{M} \\
			  &= m^\star_{\F_q}(1,2D+1,n) 
			  \leq n H_q\left(\frac{2D+1}{n}\right),
\end{align*}
where the last inequality follows from the asymptotic Gilbert-Varshamov bound, using the assumption that $\frac{2D+1}{n} \leq 1 - \frac{1}{q}$.
\end{proof}

\begin{corollary}
\label{asymptotic_existence}
Under the conditions of the previous theorem with the additional assumption that $D = o(n)$, there exist matrices $M$ that can $\Delta$-approximate $d$ with $(2D+1) \log_q \br*{\frac{n}{2D+1}} + O(D) = O(D \log (n/D))$ rows for sufficiently large $n$.
\end{corollary}
\begin{proof}
By the previous theorem, such matrices exist with $nH_q(\frac{2D+1}{n})$ rows. Expanding, we have
\begin{align*}
n H_q\left(\frac{2D+1}{n}\right) &= n\left((2D+1) \log_q(q-1) - \frac{2D+1}{n} \log_q \left(\frac{2D+1}{n}\right) - \left(1-\frac{2D+1}{n}\right) \log_q \left(1- \frac{2D+1}{n}\right)\right) \\
&= (2D+1) \log_q(q-1) - (2D+1) \log_q \left(\frac{2D+1}{n}\right) - (n - 2D - 1) \log_q \left(1-\frac{2D+1}{n}\right) \\
&= O(D) + (2D+1) \log_q \left(\frac{n}{2D+1}\right) - (n - 2D - 1) \log_q\left(1- \frac{2D+1}{n}\right),
\end{align*}
and then note that the last term goes to 0 for large $n$, as $\frac{2D+1}{n}$ goes to 0 since $D = o(n)$.
\end{proof}

The above result is asymptotically tight with the lower bound given previously, so only improvements in the constant factor are possible. For certain settings of parameters these improvements follow easily from known results; for instance, the existence of binary BCH codes with $n-k \leq D \log_2(n+1)$ for certain values of $n$ implies that when $q=2$ the bound in \cref{asymptotic_existence} improves by about a factor of 2.

\vspace{2mm} 

\noindent{\textbf{Remark: Adaptive Tests.}} In the case of linear measurements over $\real$, a simple trick yields a very efficient adaptive scheme as well. First, suppose the entries of the vector $\bfx$  are nonnegative. In this case, the result of a test is nonzero if and only if some nonzero entry of $\bfx$ is included in the test.
Then for the purpose of determining sparsity, each test tells us as least as much information as the corresponding test in the group testing model, as we can simply threshold the real-valued test result to 1 if it is nonzero and 0 otherwise. Thus we can exploit existing results for adaptive group testing sparsity approximation \cite{falahatgar2016estimating} to obtain an estimate of the sparsity that is accurate with high probability in as few as $O(\log \log d)$ measurements.
However, if not all entries of the vector $\bfx$ are nonnegative, then there is the additional complication that it is possible to observe a test result of 0 even when a nonzero entry of $\bfx$ is included in the test. This will happen exactly when the vector corresponding to some test is orthogonal to $\bfx$. It is easy to counteract this by slightly perturbing each test vector, adding a small real-valued random vector that is nonzero only on the support of the test vector to each test. This ensures that the new test vector lies in the space orthogonal to $\bfx$ with probability 0.

\bibliographystyle{abbrv}
\bibliography{sparse_estimation}

\end{document}

%% file: format.tex
\usepackage{amsmath,amsthm,mathtools}
\usepackage{amssymb}
\usepackage{mathrsfs}

\usepackage{multirow}
\usepackage{array}

\usepackage{amsfonts}
\usepackage{graphicx}
\usepackage[linesnumbered]{algorithm2e}

\usepackage{color, verbatim}
\usepackage[usenames,dvipsnames]{xcolor}
\usepackage{cleveref}

\usepackage{epsfig}
\usepackage{subfig}
\usepackage{epstopdf}
\usepackage{xargs}
\usepackage{hhline}

\newcommand\nc\newcommand
\nc{\bb}[1]{\mathbb{#1}}
\renewcommand{\cal}[1]{\mathcal{#1}}
\renewcommand{\bf}[1]{\mathbf{#1}}

\DeclarePairedDelimiter{\ceil}{\lceil}{\rceil}
\DeclarePairedDelimiter{\floor}{\lfloor}{\rfloor}
\DeclarePairedDelimiter{\set}{\lbrace}{\rbrace}
\DeclarePairedDelimiter{\br}{\lparen}{\rparen}
\DeclarePairedDelimiter{\brac}{\lbrack}{\rbrack}
\DeclarePairedDelimiter{\abs}{\lvert}{\rvert}

\nc\bfa{{\bf{a}}}\nc\bfA{{\boldsymbol A}}\nc\cA{{\cal A}} \nc\fA[1]{A\br*{#1}} \nc\fa[1]{a\br*{#1}}  \nc\rmA{\mathrm{A}} \nc\rma{\mathrm{a}}
\nc\bfb{{\bf{b}}}\nc\bfB{{\boldsymbol B}}\nc\cB{{\cal B}} \nc\fB[1]{B\br*{#1}} \nc\fb[1]{b\br*{#1}}  \nc\rmB{\mathrm{B}} \nc\rmb{\mathrm{b}}
\nc\bfc{{\bf{c}}}\nc\bfC{{\boldsymbol C}}\nc\cC{{\cal C}} \nc\fC[1]{C\br*{#1}} \nc\fc[1]{c\br*{#1}}  \nc\rmC{\mathrm{C}} \nc\rmc{\mathrm{c}}
\nc\bfd{{\bf{d}}}\nc\bfD{{\boldsymbol D}}\nc\cD{{\cal D}} \nc\fD[1]{D\br*{#1}} \nc\fd[1]{d\br*{#1}}  \nc\rmD{\mathrm{D}} \nc\rmd{\mathrm{d}}
\nc\bfe{{\bf{e}}}\nc\bfE{{\boldsymbol E}}\nc\cE{{\cal E}} \nc\fE[1]{E\br*{#1}} \nc\fe[1]{e\br*{#1}}  \nc\rmE{\mathrm{E}} \nc\rme{\mathrm{e}}
\nc\bff{{\bf{f}}}\nc\bfF{{\boldsymbol F}}\nc\cF{{\cal F}} \nc\fF[1]{F\br*{#1}} \nc\ff[1]{f\br*{#1}}  \nc\rmF{\mathrm{F}} \nc\rmf{\mathrm{f}}
\nc\bfg{{\bf{g}}}\nc\bfG{{\boldsymbol G}}\nc\cG{{\cal G}} \nc\fG[1]{G\br*{#1}} \nc\fg[1]{g\br*{#1}}  \nc\rmG{\mathrm{G}} \nc\rmg{\mathrm{g}}
\nc\bfh{{\bf{h}}}\nc\bfH{{\boldsymbol H}}\nc\cH{{\cal H}} \nc\fH[1]{H\br*{#1}} \nc\fh[1]{h\br*{#1}}  \nc\rmH{\mathrm{H}} \nc\rmh{\mathrm{h}}
\nc\bfi{{\bf{i}}}\nc\bfI{{\boldsymbol I}}\nc\cI{{\cal I}} \nc\fI[1]{I\br*{#1}} \nc\rmI{\mathrm{I}} \nc\rmi{\mathrm{i}}
\nc\bfj{{\bf{j}}}\nc\bfJ{{\boldsymbol J}}\nc\cJ{{\cal J}} \nc\fJ[1]{J\br*{#1}} \nc\fj[1]{j\br*{#1}} \nc\rmJ{\mathrm{J}} \nc\rmj{\mathrm{j}}
\nc\bfk{{\bf{k}}}\nc\bfK{{\boldsymbol K}}\nc\cK{{\cal K}} \nc\fK[1]{K\br*{#1}} \nc\fk[1]{k\br*{#1}} \nc\rmK{\mathrm{K}} \nc\rmk{\mathrm{k}}
\nc\bfl{{\bf{l}}}\nc\bfL{{\boldsymbol L}}\nc\cL{{\cal L}} \nc\fL[1]{L\br*{#1}} \nc\fl[1]{l\br*{#1}} \nc\rmL{\mathrm{L}} \nc\rml{\mathrm{l}}
\nc\bfm{{\bf{m}}}\nc\bfM{{\boldsymbol M}}\nc\cM{{\cal M}} \nc\fM[1]{M\br*{#1}} \nc\fm[1]{m\br*{#1}} \nc\rmM{\mathrm{M}} \nc\rmm{\mathrm{m}}
\nc\bfn{{\bf{n}}}\nc\bfN{{\boldsymbol N}}\nc\cN{{\cal N}} \nc\fN[1]{N\br*{#1}} \nc\fn[1]{n\br*{#1}} \nc\rmN{\mathrm{N}} \nc\rmn{\mathrm{n}}
\nc\bfo{{\bf{o}}}\nc\bfO{{\boldsymbol O}}\nc\cO{{\cal O}} \nc\fO[1]{O\br*{#1}} \nc\fo[1]{o\br*{#1}} \nc\rmO{\mathrm{O}} \nc\rmo{\mathrm{o}}
\nc\bfp{{\bf{p}}}\nc\bfP{{\boldsymbol P}}\nc\cP{{\cal P}} \nc\fP[1]{P\br*{#1}} \nc\fp[1]{p\br*{#1}} \nc\rmP{\mathrm{P}} \nc\rmp{\mathrm{p}}
\nc\bfq{{\bf{q}}}\nc\bfQ{{\boldsymbol Q}}\nc\cQ{{\cal Q}} \nc\fQ[1]{Q\br*{#1}} \nc\fq[1]{q\br*{#1}} \nc\rmQ{\mathrm{Q}} \nc\rmq{\mathrm{q}}
\nc\bfr{{\bf{r}}}\nc\bfR{{\boldsymbol R}}\nc\cR{{\cal R}} \nc\fR[1]{R\br*{#1}} \nc\fr[1]{r\br*{#1}} \nc\rmR{\mathrm{R}} \nc\rmr{\mathrm{r}}
\nc\bfs{{\bf{s}}}\nc\bfS{{\boldsymbol S}}\nc\cS{{\cal S}} \nc\fS[1]{S\br*{#1}} \nc\fs[1]{s\br*{#1}} \nc\rmS{\mathrm{S}} \nc\rms{\mathrm{s}}
\nc\bft{{\bf{t}}}\nc\bfT{{\boldsymbol T}}\nc\cT{{\cal T}} \nc\fT[1]{T\br*{#1}} \nc\ft[1]{t\br*{#1}} \nc\rmT{\mathrm{T}} \nc\rmt{\mathrm{t}}
\nc\bfu{{\bf{u}}}\nc\bfU{{\boldsymbol U}}\nc\cU{{\cal U}} \nc\fU[1]{U\br*{#1}} \nc\fu[1]{u\br*{#1}} \nc\rmU{\mathrm{U}} \nc\rmu{\mathrm{u}}
\nc\bfv{{\bf{v}}}\nc\bfV{{\boldsymbol V}}\nc\cV{{\cal V}} \nc\fV[1]{V\br*{#1}} \nc\fv[1]{v\br*{#1}} \nc\rmV{\mathrm{V}} \nc\rmv{\mathrm{v}}
\nc\bfw{{\bf{w}}}\nc\bfW{{\boldsymbol W}}\nc\cW{{\cal W}} \nc\fW[1]{W\br*{#1}} \nc\fw[1]{w\br*{#1}} \nc\rmW{\mathrm{W}} \nc\rmw{\mathrm{w}}
\nc\bfx{{\bf{x}}}\nc\bfX{{\boldsymbol X}}\nc\cX{{\cal X}} \nc\fX[1]{X\br*{#1}} \nc\fx[1]{x\br*{#1}} \nc\rmX{\mathrm{X}} \nc\rmx{\mathrm{x}}
\nc\bfy{{\bf{y}}}\nc\bfY{{\boldsymbol Y}}\nc\cY{{\cal Y}} \nc\fY[1]{Y\br*{#1}} \nc\fy[1]{y\br*{#1}} \nc\rmY{\mathrm{Y}} \nc\rmy{\mathrm{y}}
\nc\bfz{{\bf{z}}}\nc\bfZ{{\boldsymbol Z}}\nc\cZ{{\cal Z}} \nc\fZ[1]{Z\br*{#1}} \nc\fz[1]{z\br*{#1}} \nc\rmZ{\mathrm{Z}} \nc\rmz{\mathrm{z}}

\DeclareMathOperator{\Dim}{dim}
\DeclareMathOperator{\Supp}{supp}
\DeclareMathOperator{\Rank}{rank}

\nc\defeq{\coloneqq}

\newcommand{\normO}[1]{||#1||_{\ell_0}}
\newcommand{\prob}[1]{\mathbb{P}\br*{#1}}
\newcommand{\Exp}[1]{\exp\br*{#1}}
\newcommand{\supp}[1]{\Supp\br*{#1}}
\newcommand{\Log}[1]{\log\br*{#1}}
\newcommand{\rank}[1]{\Rank\br*{#1}}
\newcommand{\Min}[1]{\min\br*{#1}}

\DeclarePairedDelimiterX\Set[1]\{\}{#1}

\newcommand\real{{\mathbb R}}

\newcommand\F{{\mathbb F}}

\newtheorem{theorem}{Theorem}
\newtheorem{definition}{Definition}
\crefname{definition}{defn.}{defns}
\newtheorem{lemma}[theorem]{Lemma}

\newtheorem{corollary}[theorem]{Corollary}